\documentclass[a4paper, oneside, 11pt]{article}%
\usepackage{amsmath}
\usepackage{amsfonts}
\usepackage{amssymb}
\usepackage{graphicx}
\usepackage{appendix}
\usepackage[square, numbers, sort&compress]{natbib}
\usepackage{color}%
\providecommand{\U}[1]{\protect\rule{.1in}{.1in}}

\newtheorem{theorem}{Theorem}[section]

\newtheorem{definition}[theorem]{Definition}

\newtheorem{lemma}[theorem]{Lemma}

\newenvironment{proof}[1][Proof]{\noindent\textbf{#1.} }{\ \rule{0.5em}{0.5em}}
\numberwithin{equation}{section}

\newcommand{\E}{{\mathbb E}}
\newcommand{\R}{{\mathbb R}}

\usepackage{url}
\usepackage{hyperref}
\hypersetup{
colorlinks=true, 
breaklinks=true, 
urlcolor=green , 
linkcolor=red, 
citecolor=blue, 
pdftitle={}, 
pdfauthor={}, 
pdfsubject={} 
}

\newcommand{\argmin}{\ensuremath{\operatorname*{argmin}}}

\newcommand{\proj}{\ensuremath{\mathrm{Proj}}}

\newcommand{\dd}{\operatorname{d}\! }
\newcommand{\dt}{\operatorname{d}\! t}
\newcommand{\ds}{\operatorname{d}\! s}

\newcommand{\dy}{\operatorname{d}\! y}

\newcommand{\ddp}{\operatorname{d}\! p}

\newcommand{\dw}{\operatorname{d}\! W}

\usepackage[normalem]{ulem}

\renewcommand{\geq}{\geqslant}
\renewcommand{\leq}{\leqslant}
\newcommand{\nn}{\nonumber}

\newcommand{\lnu}{\ensuremath{L^{2,\nu}}}

\newcommand{\linnu}{\ensuremath{L^{\infty,\nu}_{\mathcal{P}}(0, T;\mathbb{R})}}

\begin{document}

\title{Constrained monotone mean--variance investment-reinsurance under the Cram\'er--Lundberg model with random coefficients}
\author{ 
Xiaomin Shi\thanks{School of Statistics and Mathematics, Shandong University of Finance and Economics, Jinan, 250100, Shandong, China. Email: \texttt{shixm@mail.sdu.edu.cn}}
\and Zuo Quan Xu\thanks{Department of Applied Mathematics, The Hong Kong Polytechnic University, Kowloon, Hong Kong, China. Email: \texttt{maxu@polyu.edu.hk}}}
\maketitle
\begin{abstract} 
This paper studies an optimal investment-reinsurance problem for an insurer (she) under the Cram\'er--Lundberg model with monotone mean--variance (MMV) criterion. At any time, the insurer can purchase reinsurance (or acquire new business) and invest in a security market consisting of a risk-free asset and multiple risky assets whose excess return rate and volatility rate are allowed to be random. The trading strategy is subject to a general convex cone constraint, encompassing no-shorting constraint as a special case. The optimal investment-reinsurance strategy and optimal value for the MMV problem are deduced by solving certain backward stochastic differential equations with jumps. In the literature, it is known that models with MMV criterion and mean--variance criterion lead to the same optimal strategy and optimal value when the wealth process is continuous. Our result shows that the conclusion remains true even if the wealth process has compensated Poisson jumps and the market coefficients are random.
\end{abstract}
\bigskip


\section{Introduction}

In order to tame the drawback of non-monotonicity of the celebrated
mean--variance (MV) portfolio selection theory, Maccheroni et al. ~\cite{MMRT}
propose and solve, in a single-period setting, the monotone mean--variance
(MMV) model in which the objective functional is the best approximation of the
MV functional among those which are monotone.
Trybu{\l}a and Zawisza~\cite{TZ} consider the MMV problem in a stochastic
factor model using the
Hamilton--Jacobi-Bellman-Isaacs (HJBI) equations approach. With conic convex
constraints on the portfolios, Shen and Zou~\cite{SZ}, Hu, Shi and
Xu~\cite{HSX} solve the MMV problem with deterministic and random coefficients
by means of the HJBI equation approach and backward stochastic differential
equation (BSDE) approach.
All the obtained optimal investment strategy and the optimal value
in~\cite{HSX,SZ,TZ} coincide with that for the classical MV problem.

In a rather different way,
when the underlying asset prices are continuous, or even with general trading
constraints, Strub and Li~\cite{SL}, Du and Strub~\cite{DS23} prove directly
that the terminal wealth levels corresponding to the optimal portfolio
strategies for both the MMV and the MV problems will drop in the domain of the
monotonicity of the classical MV functional. This, with earlier results
in~\cite{MMRT}, leads to the conclusion that the optimal strategies for the MMV
and the MV problems always coincide when the underlying asset prices are
continuous. Hence further research on the MMV problems should focus on models
where asset prices are not continuous as claimed in~\cite[Conlusions]{DS23}.

Li, Liang, Pang~\cite{LLP23} solve,
in a jump diffusion factor model, hence with discontinuous wealth dynamics, the
MMV problem by HJBI equation method. But the results in~\cite{LLP23} seem not
correct since the Dol\'eans-Dade stochastic exponential is not required to be
positive. In the Cram\'er-Lundberg risk model, a particular jump diffusion
scenario,
Li, Guo and Tian~\cite{LGT}, Li, Liang, Pang~\cite{LLP} study the optimal
investment-reinsurance problems under the criteria of MMV and MV respectively
and indicate that solutions to these two problems coincide as well. In their
model, all the coefficients are deterministic, so the problem is reduced to
solve a partial differential equation.

{In practice, however, the market parameters, such as the interest rate,
stock appreciation and volatility rates are affected by uncertainties caused by
the Brownian motion and Poisson random measure. Thus, it is too restrictive to
set market parameters as deterministic functions of $t$. On the other
hand, in real financial market, especially for an insure, short-selling some or
all the stocks are prohibited. Thus we aim} to generalize the models
in~\cite{LGT,LLP} to allowing random coefficients (both excess
return rates and volatility rate) and convex cone trading constraints. We
follow the same idea of our previous work~\cite{HSX2} without jumps: firstly
guess an optimal portfolio candidate, via a heuristic argument, by means of
specific BSDEs, and then prove a verification theorem rigorously. Along this
line, we eventually provide a {semi-closed-form} solution for the
constrained MMV investment-reinsurance problem in terms of some BSDE with
jumps. Unlike models without jumps, $\psi$ in our model is required to be
no less than $-1+\epsilon$ to guarantee that $\Lambda^{\eta,\psi}$ is a strictly positive
martingale. Properties of the corresponding BSDE help to fit exactly this
subtle point. 

{It is worth pointing out that, in our another preprint~\cite{SX}, we
study the optimal investment-reinsurance problems under the MV criterion, with
the same random coefficients and cone constraints setting. Explicit MV optimal
investment-reinsurance strategies in terms of a partially coupled BSDEs with
jumps (i.e.~\eqref{P1}, \eqref{P2}) are provided, and the existence and
uniqueness of solutions to the corresponding BSDEs are established, which are
the main theoretical contributions of~\cite{SX}.}
Then we compare carefully the optimal strategies for MMV problem and MV
problems obtained in the current paper and~\cite{SX} respectively,
and reach the same conclusion as all the existing models, that is, the MMV and MV {criteria} lead to the same optimal strategy.

The rest part of this paper is organized as follows.
In Section \ref{fm}, we present the financial market and formulate the constrained MMV problem with random coefficients. In Section \ref{Heuri}, we derive heuristically the optimal candidate. Section \ref{Verification} provides a rigorous verification for the optimal investment-reinsurance strategy and optimal value. In Section \ref{Comparison}, we solve the constrained MV problem with random coefficients and make a comparison with the MMV problem. Some concluding remarks are given in Section \ref{conclude}.
\section{Problem formulation}
\label{fm}

Let $(\Omega, \mathcal F, \mathbb{P})$ be a fixed complete probability space on which are defined a standard $n$-dimensional Brownian motion $W$. We denote by $\R^m$ the set of $m$-dimensional column vectors, by $\R^m_+$ the set of vectors in $\R^m$ whose components are nonnegative, by $\R^{m\times n}$ the set of $m\times n$ real matrices, and by $\mathbb{S}^n$ the set of symmetric $n\times n$ real matrices. 
For $M=(m_{ij})\in \R^{m\times n}$, we denote its transpose by $M'$, and its norm by $\vert M\vert =\sqrt{\sum_{ij}m_{ij}^2}$. If $M\in\mathbb{S}^n$ is positive definite (resp. positive semidefinite), we write $M>$ (resp. $\geq$) $0$. We write $A>$ (resp. $\geq$) $B$ if $A, B\in\mathbb{S}^n$ and $A-B>$ (resp. $\geq$) $0$.
As usual, we write $x^+=\max\{x, 0\}$ and $x^-=\max\{-x, 0\}$ for $x\in\R$.

\subsection{The insurance and financial model}

Let
$\{N_t\}$ be a homogeneous Poisson process with intensity $\lambda>0$ which counts the number of claims. Let $Y_i$ denote the payment amount of the $i^{\mathrm{th}}$ claim, $i=1,2, \ldots $, and assume they are bounded 
independent and identically distributed nonnegative random variables with a common probability distribution function $\nu:\R_+\rightarrow[0,1]$.
{In insurance practice, an insurance claim shall never exceed the value of the insured asset,
which is usually upper bounded by some constant (such as the value of a new replacement of the insured car or house). Even if the insured asset (such as human life) is invaluable, there is usually a maximum payment amount in practice. Therefore, it is reasonable to assume the claims $Y_i$ are bounded.} Then the first and second moments of the claims exist, denoted by
\begin{equation}
\label{def:bY}
b_Y=\int_{\R_+} y \nu(\dy)>0,\quad \sigma_Y^2=\int_{\R_+} y^2 \nu(\dy)>0.
\end{equation}
In the classical Cram\'er-Lundberg model, the surplus $R_t$ dynamic without reinsurance or investment follows
\begin{equation}
\label{Rprocess}
R_t=R_0+pt-\sum_{i=1}^{N_t} Y_i, 
\end{equation}
where $p$ is the premium rate which is assumed to be calculated
according to the expected value principle, i.e.~$p=(1+\eta)\lambda b_Y>0$ and $\eta>0$ is
the relative safety loading of the insurer. As in~\cite{SZ2}, we use a Poisson
random measure $\gamma$ on $[0,T]\times\Omega\times\R_+$ to represent the compound Poisson
process $\sum_{i=1}^{N_t} Y_i$ as\footnote{{This method was also used to study a more
general process in Elliott, Siu, Yang~\cite{ESY09}.}} 
\[
\int_0^t\int_{\R_+}y\gamma(\ds,\dy)=\sum_{i=1}^{N_t} Y_i.
\]
Assume that $N, Y_i, i=1,2, \ldots $ are independent, then
\[
\E\Bigl[\sum_{i=1}^{N_t} Y_i\Bigr]=\E[N_t]\E[Y_1]=\lambda t\int_{\R_+}y\nu(\dy).
\]
The compensated Poisson random measure is denoted by 
\[
\tilde \gamma(\dt,\dy)=\gamma(\dt,\dy)-\lambda\nu(\dy)\dt.
\]

We suppose that the Brownian motion $W$ and the Poisson random measure $\gamma(\dt,\dy)$ are independent processes {under the probability measure $\mathbb{P}$}. Define the filtration $\mathbb{F}=\{\mathcal{F}_t,t\geq0\}$ as the augmented natural filtration {generated by} $W$ and $\gamma$.

Let $T>0$ be a constant that will stand for the investment horizon throughout the paper.
Let $\mathcal{P}$ be the $\mathbb{F}$-predictable $\sigma$-field on $[0,T]\times\Omega$, $\mathcal{B}(\mathbb{R}_+)$ the Borel $\sigma$-algebra of $\mathbb{R}_+$.

We now introduce several spaces. 
Let $L^{2}_{\mathbb F}(0, T;\mathbb{R})$ be the set of $\mathcal{P}$-measurable functions $\varphi:[0, T]\times\Omega\to\mathbb{R}$ such that $\E\int_0^T\vert \varphi\vert ^2\dt<\infty$. 
Let $\lnu$ be the set of $\mathcal{B}(\R_+)$-measurable functions $\varphi :\R_+\rightarrow\mathbb{R}$ such that $\int_{\R_+} \varphi(y)^2\nu(\dy)<\infty$.
Let $L^{2,\nu}_{\mathcal{P}}(0, T;\mathbb{R})$ be the set of $\mathcal{P}\otimes\mathcal{B}(\R_+)$-measurable functions $\varphi :[0, T]\times\Omega\times \R_+\rightarrow
\mathbb{R}$ such that $\E\int_0^T\int_{\R_+}\vert \phi\vert ^2\lambda\nu(\dy)\dt<\infty$. 
Let $L^{\infty,\nu}_{\mathcal{P}}(0, T;\mathbb{R})$ be the set of functions $\varphi\in L^{2,\nu}_{\mathcal{P}}(0, T;\mathbb{R})$ which are essentially bounded w.r.t. $\dt\otimes \dd\mathbb{P}\otimes \nu(\dy)$.
Let $S^{\infty}_{\mathbb{F}}(0,T;\mathbb{R})$ be the set of $\mathbb{F}$-adapted functions $\varphi :[0,T]\times\Omega\rightarrow\mathbb{R}
$ which are c\`ad-l\`ag and essentially bounded w.r.t. $\dt\otimes \dd\mathbb{P}$. 
The above definitions are generalized in the obvious way to the cases that $\mathbb{R}$ is replaced by $\mathbb{R}^n$, $\mathbb{R}^{n\times m}$ or $\mathbb{S}^n$.
In our argument, $s$, $t$, $\omega$, ``almost surely'' and ``almost everywhere'', will be suppressed for simplicity in many circumstances, when no confusion occurs.

We consider an insurer (she) who is allowed to purchase reinsurances or acquire new businesses to control its exposure to the insurance risk. Let $q_t$ be the value of risk exposure, which represents the insurer's retention level of insurance risk at time $t$. When $q_t\in[0,1]$, it corresponds to a proportional reinsurance cover; in this case, the insurer only need pays $100 q_t\%$ of each claim, and the reinsurer pays the rest;
as a remedy,
the insurer diverts part of the premium to the reinsurer at the rate of $(1-q_t)(1+\eta_r)\lambda b_Y$, where $\eta_r\geq\eta$ is the reinsurer's relative security loading. When $q_t>1$, it corresponds to acquiring new business. The process $q_t$ is called a reinsurance strategy for convenience. By adopting a reinsurance strategy $q_t$, the insurer's surplus process $R^q$ follows
\begin{align*}
\dd R^q_t&=(1+\eta)\lambda b_{Y}\dt-(1-q_t)(1+\eta_r)\lambda b_{Y}\dt -q_t\int_{\R_+}y\gamma(\dt,\dy)\\
&=(\eta_r q_t+\eta-\eta_r)\lambda b_Y\dt-q_t\int_{\R_+}y\tilde \gamma(\dt,\dy).
\end{align*}

The insurer is allowed to invest her wealth in a financial market consisting of a risk-free asset (the money market
instrument or bond) whose price is $S_{0}$ and $m$ risky securities (the
stocks) whose prices are $S_{1}, \ldots, S_{m}$. And assume $m\leq n$, i.e.~the number of risky securities is no more than the dimension of the Brownian motion {(part of the source of market randomness)}. {In this case, there always have randomness that cannot be perfectly hedged. Hence this is an incomplete market.}
The asset prices $S_k$, $k=0, 1, \ldots, m$, are driven by stochastic differential equations (SDEs):
\[
\begin{cases}
\dd S_{0,t}=r_tS_{0,t}\dt, \\
S_{0,0}=s_0,
\end{cases}\]
and
\[
\begin{cases}
\dd S_{k,t}=S_{k,t}\Bigl((\mu_{k,t}+r_t)\dt+\displaystyle \sum_{j=1}^n\sigma_{kj,t}\dw_{j,t}\Bigr), \\
S_{k,0}=s_k,
\end{cases}\]
where, for every $k=1, \ldots, m$, $r$ is the interest rate process, $\mu_k$ and $\sigma_k:= (\sigma_{k1}, \ldots, \sigma_{kn})$ are the mean excess return rate process and volatility rate process of the $k$th risky security.
Define the mean excess return vector
$\mu=(\mu_1, \ldots, \mu_m)^{\top}$,
and the volatility matrix $\sigma= (\sigma_{kj})_{m\times n}$. 
Unlike most existing models (e.g.~Li, Liang, Pang~\cite{LLP23}), we allow them
to be random, {non-anticipative with respect to the filtration
$\mathbb{F}$}.
Precisely, we assume $\mu\in L_{\mathbb{F}}^\infty(0, T;\mathbb{R}^m)$ and
$\sigma\in L_{\mathbb{F}}^\infty(0, T;\mathbb{R}^{m\times n})$. Furthermore, there exists a constant $\delta>0$ such that
$\sigma\sigma'\geq \delta 1_m$ for a.e. $t\in[0, T]$, where $1_m$ denotes the $m$-dimensional identity matrix.
Same as most of existing models, we assume the interest rate $r$ is a
bounded deterministic measurable function of $t$. It is a
long-standing problem how to solve cone constrained MV problem when the
interest rate $r$ is random. {Actually, with random interest rate
$r$, it is very difficult to construct one (or even more are required)
auxiliary adapted process to recover the homogeneity, a critical property used
in~\cite{HSX,HZ} for studying cone constrained stochastic control problems. In
the cone constrained MMV problem~\cite{HSX2}, the main difficulty in carrying
out the construction of $R^{(\eta,\psi,\pi,q)}$ could not be avoided in this paper if
$r$ is random. }

We assume the insurer is a small investor so that her actions cannot affect the asset prices. She can decide at every time
$t\in[0, T]$ the amount $\pi_{j,t}$ of her wealth to invest in the $j$th risky asset, $j=1, \ldots, m$. The vector process $\pi:= (\pi_1, \ldots, \pi_m)'$ is called a portfolio of the investor. Then the investor's self-financing wealth process $X$ corresponding to an (investment-reinsurance) strategy $(\pi,q)$ satisfies the following SDE:
\begin{equation}
\label{wealth}
\begin{cases} 
\dd X_t=[r_tX_{t-}+\pi_t'\sigma_t\phi_t+bq_t+a]\dt +\pi_t'\sigma_t\dw_t-q_t\int_{\R_+}y\tilde \gamma(\dt,\dy), \\
\; X_0=x,
\end{cases}\end{equation}
where
\[\phi:= \sigma'(\sigma\sigma)^{-1}\mu, ~~b:= \lambda b_Y\eta_r>0,~~ a:= \lambda b_Y(\eta-\eta_r).
\]

We may write $X^{\pi,q} $ instead of $X$ to emphasis the dependent of $X$ on the strategy $(\pi,q)$ in \eqref{wealth}.

Let $\Pi$ be a given closed convex cone in $\mathbb{R}^m$, {i.e.,~$\Pi$ is closed, convex, and if $u\in\Pi$, then $c u\in\Pi$, for all $c\geq 0$.} 
It is the constraint set for investment. { Note that both $\Pi=\R^m_+$ and $\Pi=\R^{m_0}_+\times\R^{m-m_0}$ ($m_0< m$) are exactly closed convex cones. And the former means no-shorting constraint; while the later means that shorting-selling the first $m_0$ stocks are prohibited.}
The class of admissible investment-reinsurance strategies is defined as
\[
\mathcal{U}:= \Bigl\{(\pi,q)\in L^2_{\mathbb{F}}(0, T;\mathbb{R}^{m+1})\;\Big\vert \; \pi \in\Pi, \ q\geq0\Bigr\}.
\]
For any admissible strategy $(\pi,q)\in\mathcal{U}$, the wealth process \eqref{wealth} admits a unique strong solution $X^{\pi,q}$.

\subsection{A family of probability measures}

{For any $(\eta,\psi)\in L^2_{\mathcal{F}}(0, T;\mathbb{R}^m)\times L^{2,\nu}_{\mathcal{P}}(0,T;\R)$, the process 
\[
M_t:= \int_0^t\eta_s'\dw_s+\int_0^t\int_{\R_+}\psi_s(y)\tilde\gamma(\ds,\dy)\]
is a square integrable martingale. According to~\cite[Theorem 10.9]{HWY}, the
above process $M_t$ is a BMO martingale if and only if there exists a
constant $c>0$ such that
\begin{align}
&\E\Bigl[\int_\tau^T\Bigl(\vert \eta_s\vert ^2+\int_{\R_+}\vert \psi_{s}(y)\vert ^2\nu(\dy)\Bigr)\ds\Big\vert \mathcal{F}_{\tau}\Bigr]\leq c,\label{BMO1}~~~\vert M_{\tau}-M_{\tau-}\vert ^2\leq c, 
\end{align}
for any $\mathbb{F}$-stopping times $\tau\leq T$.
}
{
The following result concerning the Dol\'eans-Dade stochastic exponential of a
BMO martingale, can be found in~\cite{Ka}.}
\begin{lemma}[Kazamaki's Criterion]
Let $M$ be a BMO martingale such that there exists $\varepsilon>0$ with $M_{t}-M_{t-}\geq -1+\varepsilon$, for all $t\in[0,T]$, $\mathbb{P}$-a.s., then $\mathcal{E}(M)$ is a strictly positive martingale.
\end{lemma}

{Denote
\begin{align}\label{def:A}
\mathcal{A} := \bigcup_{\epsilon>0}\Bigl\{ &(\eta,\psi)\in L^2_{\mathcal{F}}(0, T;\mathbb{R}^m)\times L^{\infty,\nu}_{\mathcal{P}}(0,T;\R)\nn\\
 &\Big\vert \; (\eta,\psi) \ \mbox{satisfies} \ \text{\eqref{BMO1}},~ \psi\geq -1+\epsilon,~\E[(\Lambda_T^{\eta,\psi})^{2}]<\infty \Bigr\}.
\end{align} 
Then for any processes $(\eta,\psi)\in\mathcal{A}$, the Kazamaki's criterion clearly holds. Therefore, we can define a probability $\mathbb{P}^{\eta,\psi}$ through
\[
\frac{d\mathbb{P}^{\eta,\psi}}{d\mathbb{P}}\bigg\vert _{\mathcal{F}_t}=\Lambda^{\eta,\psi}_t, \ t\in[0,T],\]
where the Dol\'eans-Dade stochastic exponential 
\[
\Lambda^{\eta,\psi}:= \mathcal{E}\left(\int_0^t\eta_s'\dw_s+\int_0^t\int_{\R_+}\psi_s(y)\tilde\gamma(\ds,\dy)\right)\]
is a strictly positive martingale on $[0,T]$.

By Girsanov's theorem, 
\[
W^\eta_t:= W_t-\int_0^t\eta_s\ds, \ t\in[0,T]\]
is a Brownian motion, and 
\[
\tilde\gamma^{\psi}(\dt,\dy):= \tilde\gamma(\dt,\dy)-\lambda\psi_t(y)\nu(\dy), \ t\in[0,T]\]
is a compensated Poisson random measure under $\mathbb{P}^{\eta,\psi}$.}

\subsection{The monotone mean--variance problem}

According to~\cite[Page 489]{MMRT}, the MMV functional, defined by 
\[
\inf_{(\eta,\psi)\in\mathcal{A}}\E^{\mathbb{P}^{\eta,\psi}}
\Bigl[X_T^{\pi,q}+\frac{1}{2\theta}(\Lambda^{\eta,\psi}_T-1)\Bigr]\]
is the minimal monotone functional that dominates the MV functional $$\E^{\mathbb{P}}(X_T^{\pi,q})-\frac{\theta}{2}Var^{\mathbb{P}}(X_T^{\pi,q}),$$ where $\theta$ is a given positive constant {measuring the risk aversion of the insurer.}

In this paper, we consider the following MMV problem:
\begin{equation}
\label{MMV}
\sup_{(\pi,q)\in\mathcal{U}}\inf_{(\eta,\psi)\in\mathcal{A}}\E^{\mathbb{P}^{\eta,\psi}}\Bigl[X_T^{\pi,q}+\frac{1}{2\theta}(\Lambda^{\eta,\psi}_T-1)\Bigr],
\end{equation}
{The last condition in \eqref{def:A} ensures that the expectation in \eqref{MMV} is finite.}

Without loss of generality, we may assume $a=0$ in \eqref{wealth}. Indeed, by setting $\widetilde X_t=X_t+a\int_t^Te^{-\int_t^sr_{\alpha}\dd\alpha}\ds$, the optimization problem in \eqref{MMV} will not change and all the results in this paper remain true with $X_t$ replaced by $\widetilde X_t$.

\section{Optimal candidate: a heuristic derivation}

\label{Heuri}
In order to solve the problem \eqref{MMV}, we hope, via a heuristic argument, to find a family of stochastic process $R^{(\eta,\psi,\pi,q)}$ and a quaternion $(\hat\eta,\hat\psi,\hat\pi,\hat q)$ with the following properties:
\begin{enumerate}
\item $R^{(\eta,\psi,\pi,q)}_T=X_T^{\pi,q}+\frac{1}{2\theta}(\Lambda^{\eta,\psi}_T-1)$ for all $(\eta,\psi,\pi,q)\in\mathcal{A}\times\mathcal{U}$.

\item $R^{(\eta,\psi,\pi,q)}_0=R_0$ is constant for all $(\eta,\psi,\pi,q)\in\mathcal{A}\times\mathcal{U}$.

\item $\E^{\mathbb{P}^{\hat\eta,\hat\psi}}\Bigl[X_T^{\pi,q}+\frac{1}{2\theta}(\Lambda^{\hat\eta,\hat\psi}_T-1)\Bigr]
\leq R_0$ for all $(\pi,q)\in\mathcal{U}$.

\item $\E^{\mathbb{P}^{\eta,\psi}}\Bigl[X_T^{\hat\pi,\hat q}+\frac{1}{2\theta}(\Lambda^{\eta,\psi}_T-1)\Bigr]\geq R_0$ for all $(\eta,\psi)\in\mathcal{A}$.

\end{enumerate} 
If this is done, we will then rigorously show that $(\hat\eta,\hat\psi,\hat\pi,\hat q)$ is an optimal solution for the problem \eqref{MMV} and $R_0$ is its optimal value.
Clearly, the third and fourth properties ensure that $(\hat\eta,\hat\psi,\hat\pi,\hat q)$ is a saddle point for \eqref{MMV} with the optimal value 
\[
\E^{\mathbb{P}^{\hat\eta,\hat\psi}}\Bigl[X_T^{\hat\pi,\hat q}+\frac{1}{2\theta}(\Lambda^{\hat\eta,\hat\psi}_T-1)\Bigr]=R_0.
\]

We consider the following family:
\[ 
R^{(\eta,\psi,\pi,q)}_t=X_t^{\pi,q}h_t+\frac{1}{2\theta}(\Lambda_t^{\eta,\psi} Y_t-1),\]
where $(Y,Z,V)$ and $(h,L,\Phi)$ satisfy the following BSDEs with jumps, respectively, 
{(we shall often suppress the argument $t$ in BSDEs for notation simplicity)}
\begin{equation}
\label{Y0}
\begin{cases}
\dd Y_t=-f\dt+Z'\dw+\int_{\R_+}V(y)\tilde \gamma(\dt,\dy),\\
Y_T=1,\quad Y>0, \quad Y+V>0,
\end{cases}\end{equation}
and
\begin{equation}
\label{h0}
\begin{cases}
\dd h_t=-g\dt+L'\dw+\int_{\R_+}\Phi(y)\tilde \gamma(\dt,\dy),\\
h_T=1.
\end{cases}\end{equation} 
Our aim reduces to finding proper drivers $f$ and $g$, which are independent of $(\eta,\psi,\pi,q)$, such that the desired properties hold. 

Notice the first property is already satisfied by the choice of {terminal conditions $Y_{T}=h_T=1$ in the above BSDEs with jumps}.
Because $Y$ and $h$ are independent of $(\eta,\psi,\pi,q)$, it follows
\begin{align*}
R^{(\eta,\psi,\pi,q)}_0&=X_0^{\pi,q}h_0+\frac{1}{2\theta}(\Lambda_0^{\eta,\psi} Y_0-1) =xh_0+\frac{1}{2\theta}(Y_0-1):= R_0
\end{align*}
is a constant. So the second property follows.

Now rewrite the BSDEs \eqref{Y0} and \eqref{h0} in terms of $W^\eta$ and $\tilde\gamma^{\psi}$ as
\[ 
\begin{cases}
\dd Y_t=(-f+Z'\eta+\int_{\R_+}V(y)\psi(y)\lambda\nu(\dy))\dt +Z'\dw^\eta+\int_{\R_+}V(y)\tilde\gamma^{\psi}(\dt,\dy),\\
Y_T=1,\quad Y>0, \quad Y+V>0,
\end{cases}\]
and
\begin{equation}
\label{defh}
\begin{cases}
\dd h_t=(-g+L'\eta+\int_{\R_+}\Phi(y)\psi(y)\lambda\nu(\dy))\dt +L'\dw^\eta+\int_{\R_+}\Phi\tilde\gamma^{\psi}(\dt,\dy),\\
h_T=1.
\end{cases}\end{equation} 
Applying It\^o's formula we obtain
\begin{align}\label{R2}
\dd R^{(\eta,\psi,\pi,q)}_t 
&=\Bigl[\frac{\Lambda}{2\theta}\Bigl(Y\vert \eta\vert ^2+2 \eta'Z+\frac{2\theta}{\Lambda}h\eta'\sigma'\pi\Bigr)
\nn\\
&\quad\;+\frac{\Lambda}{2\theta}\int_{\R_+}\Bigl((Y+V)\psi^2+2V\psi
-\frac{2\theta}{\Lambda}hqy\psi\Bigr)\lambda\nu(\dy)\nn\\
&\quad\;+\Bigl(rh-g+L'\eta +\int_{\R_+}\Phi\psi\lambda\nu(\dy)\Bigr)X +\pi'(h\sigma\phi+\sigma L)\nn\\
&\quad\;+q\Bigl(hb-\int_{\R_+}y\Phi(1+\psi)\lambda\nu(\dy)\Bigr)-\frac{\Lambda}{2\theta}f\Bigr]\dt\nn\\
&~~+(\cdots)\dw^{\eta}+\int_{\R_+}(\cdots)\tilde\gamma^{\psi}(\dt,\dy).
\end{align} 
Since $r$ is deterministic function of $t$, we see
\begin{equation}
\label{h}
g=rh, ~ h_t=e^{\int_t^Tr_s\ds}, \ L_t=0, \ \Phi_t(y)=0,
\end{equation}
fulfills \eqref{defh}. From now on, we fix this choice, under which \eqref{R2} becomes
\begin{align*}
\dd R^{(\eta,\psi,\pi,q)}_t 
&= \biggl\{\frac{\Lambda Y}{2\theta}\Big\vert \eta+\frac{1}{Y}\Bigl(Z+\frac{\theta h}{\Lambda}\sigma'\pi\Bigr)\Big\vert ^2-\frac{\Lambda Y}{2\theta}\frac{1}{ Y^2}\Big\vert \Bigl(Z+\frac{\theta h}{\Lambda}\sigma'\pi\Bigr)\Big\vert ^2\\
&\qquad+\frac{\Lambda }{2\theta}\int_{\R_+}\Bigl[(Y+V)\Big\vert \psi+\frac{1}{Y+V}\Bigl(V-\frac{\theta}{\Lambda}hqy\Bigr)\Big\vert ^2\nn\\
&\qquad-\frac{1}{Y+V}\Big\vert V-\frac{\theta}{\Lambda}hqy\Big\vert ^2\Bigr]\lambda\nu(\dy) +\pi'h\sigma\phi+qhb-\frac{\Lambda}{2\theta}f\biggr\}\dt\nn\\
&~~+(\cdots)\dw^{\eta}+\int_{\R_+}(\cdots)\tilde\gamma^{\psi}(\dt,\dy).
\end{align*}
Integrating from $0$ to $T$, and taking expectation $\E^{\mathbb{P}^{\eta,\psi}}$, we have
\begin{align}
\E^{\mathbb{P}^{\eta,\psi}}[R^{(\eta,\psi,\pi,q)}_{T}]
=R_0&+\E^{\mathbb{P}^{\eta,\psi}}\int_0^{T} \frac{\Lambda }{2\theta}\biggm\{ Y\Big\vert \eta
+\frac{1}{ Y}\Bigl(Z+\frac{\theta}{\Lambda} h\sigma'\pi\Bigr)\Big\vert ^2\nonumber\\
&+ \int_{\R_+}(Y+V)\Big\vert \psi+\frac{1}{Y+V}\Bigl(V-\frac{\theta}{\Lambda}hqy\Bigr)\Big\vert ^2\lambda\nu(\dy)
\nn\\
&\;-\Bigm[f+\frac{1}{Y}\Big\vert Z+\frac{\theta}{\Lambda} h\sigma'\pi\Big\vert ^2+\int_{\R_+}\frac{1}{Y+V}\Big\vert V-\frac{\theta}{\Lambda}hqy\Big\vert ^2\lambda\nu(\dy)\nn\\
&-\frac{2\theta}{\Lambda} h\pi'\sigma\phi-\frac{2\theta}{\Lambda} hqb\Bigm]\biggm\}\ds.\label{exp}
\end{align}

To fulfill the third and fourth properties,
since $\theta, \Lambda>0$, {it suffices to} find $f$, $\hat\eta, \hat\psi$, $\hat\pi$ and $\hat q$ such that
\begin{align}\label{ineq2}
&Y\Big\vert \hat\eta+\frac{1}{ Y}\Bigl(Z+\frac{\theta}{\Lambda} h\sigma'\pi\Bigr)\Big\vert ^2 +\int_{\R_+}(Y+V)\Big\vert \hat\psi+\frac{1}{Y+V}\Bigl(V-\frac{\theta}{\Lambda}hqy\Bigr)\Big\vert ^2\lambda\nu(\dy)\nn\\
&\quad-\Bigm[f+\frac{1}{Y}\Big\vert Z+\frac{\theta}{\Lambda} h\sigma'\pi\Big\vert ^2+\int_{\R_+}\frac{1}{Y+V}\Big\vert V-\frac{\theta}{\Lambda}hqy\Big\vert ^2\lambda\nu(\dy) -\frac{2\theta}{\Lambda} h\pi'\sigma\phi-\frac{2\theta}{\Lambda} hqb\Bigm]\leq 0
\end{align}
for all $\pi,q$;
\begin{align}\label{ineq1}
&Y\Big\vert \eta+\frac{1}{ Y}\Bigl(Z+\frac{\theta}{\Lambda} h\sigma'\hat\pi\Bigr)\Big\vert ^2 +\int_{\R_+}(Y+V)\Big\vert \psi+\frac{1}{Y+V}\Bigl(V-\frac{\theta}{\Lambda}h\hat qy\Bigr)\Big\vert ^2\lambda\nu(\dy)\nn\\
&-\Bigl[f+\frac{1}{Y}\Big\vert Z+\frac{\theta}{\Lambda} h\sigma'\hat\pi\Big\vert ^2+\int_{\R_+}\frac{1}{Y+V}\Big\vert V-\frac{\theta}{\Lambda}h\hat qy\Big\vert ^2\lambda\nu(\dy) -\frac{2\theta}{\Lambda} h\hat\pi'\sigma\phi-\frac{2\theta}{\Lambda} h\hat qb\Bigr]\geq 0
\end{align}
for all $\eta,\psi$.
The above two estimates clearly imply
\begin{align*}
& Y\Big\vert \hat\eta+\frac{1}{ Y}\Bigl(Z+\frac{\theta}{\Lambda} h\sigma'\hat\pi\Bigr)\Big\vert ^2 +\int_{\R_+}(Y+V)\Big\vert \hat\psi+\frac{1}{Y+V}\Bigl(V-\frac{\theta}{\Lambda}h\hat qy\Bigr)\Big\vert ^2\lambda\nu(\dy)\nn\\
& -\Bigl[f+\frac{1}{Y}\Big\vert Z+\frac{\theta}{\Lambda} h\sigma'\hat\pi\Big\vert ^2+\int_{\R_+}\frac{1}{Y+V}\Big\vert V-\frac{\theta}{\Lambda}h\hat qy\Big\vert ^2\lambda\nu(\dy) -\frac{2\theta}{\Lambda} h\hat\pi'\sigma\phi-\frac{2\theta}{\Lambda} h\hat qb\Bigr]= 0,
\end{align*}
which motives us to take
\[
\hat\eta=-\frac{1}{ Y}\Bigl(Z+\frac{\theta}{\Lambda} h\sigma'\hat\pi\Bigr), ~
\hat\psi=-\frac{1}{Y+V}\Bigl(V-\frac{\theta}{\Lambda}h\hat qy\Bigr)\]
and
\begin{align}
\label{expressf1}
f =&-\frac{1}{Y}\Big\vert Z+\frac{\theta}{\Lambda} h\sigma'\hat\pi\Big\vert ^2
+\frac{2\theta}{\Lambda} h\hat\pi'\sigma\phi+\frac{2\theta}{\Lambda} h\hat qb - \int_{\R_+}\frac{1}{Y+V}\Big\vert V-\frac{\theta}{\Lambda}h\hat qy\Big\vert ^2\lambda\nu(\dy).
\end{align}
To find proper $\hat\pi$ and $\hat q$, substituting these expressions into \eqref {ineq2}, we get an equivalent condition of \eqref{ineq2}:
\begin{align}
&\quad\;Y\Big\vert -\frac{1}{ Y}\Bigl(Z+\frac{\theta}{\Lambda} h\sigma'\hat\pi\Bigr)+\frac{1}{ Y}\Bigl(Z+\frac{\theta}{\Lambda} h\sigma'\pi\Bigr)\Big\vert ^2\nonumber\\
&\quad\;~+\int_{\R_+}(Y+V)\Big\vert -\frac{1}{Y+V}\Bigl(V-\frac{\theta}{\Lambda}h\hat qy\Bigr) +\frac{1}{Y+V}\Bigl(V-\frac{\theta}{\Lambda}h qy\Bigr)\Big\vert ^2\lambda\nu(\dy)\nn\\
& \leq -\frac{1}{Y}\Big\vert Z+\frac{\theta}{\Lambda} h\sigma'\hat\pi\Big\vert ^2+\frac{1}{Y}\Big\vert Z+\frac{\theta}{\Lambda} h\sigma'\pi\Big\vert ^2\nn\\
&\quad\;~-\int_{\R_+}\frac{1}{Y+V}\Big\vert V-\frac{\theta}{\Lambda}h\hat qy\Big\vert ^2\lambda\nu(\dy) +\int_{\R_+}\frac{1}{Y+V}\Big\vert V-\frac{\theta}{\Lambda}hqy\Big\vert ^2\lambda\nu(\dy)\nn\\
&\quad\;~+\frac{2\theta}{\Lambda} h\hat\pi'\sigma\phi-\frac{2\theta}{\Lambda} h\pi'\sigma\phi
+\frac{2\theta}{\Lambda} h\hat qb-\frac{2\theta}{\Lambda} hqb.
\end{align}
Write $u=\frac{\theta}{\Lambda} h\sigma'\pi$, $\xi=\frac{\theta}{\Lambda} h\sigma'\hat\pi$,
$\alpha=\frac{\theta}{\Lambda} h q$ and $\rho=\frac{\theta}{\Lambda} h\hat q$, then the above becomes
\begin{align}\label{projcon}
&(Y\phi-Z-\xi)'(u-\xi)+Y(\alpha-\rho) \Bigl[\int_{\R_+}\frac{V}{Y+V}y\lambda\nu(\dy)+b
-\rho\int_{\R_+}\frac{1}{Y+V}y^2\lambda\nu(\dy)\Bigr]\leq 0.
\end{align}
Because $\theta$, $\Lambda$, $h>0$, we have $q\in\R^{+}$ if and only if so is $\alpha$.
Because $\Pi$ is a cone and $\pi\in \Pi$, we see $u\in \sigma'\Pi $, 
where, for every $(t,\omega)$, the set $\sigma_t(\omega)'\Pi$ is a closed convex cone, defined as $\sigma_t(\omega)'\Pi=\bigl\{\sigma_t(\omega)'\pi\;\big\vert \;\pi\in\Pi\bigr\}$.

Now we conclude \eqref{projcon} holds for all $u\in \sigma'\Pi $, $\alpha\in\R^{+}$ if
\begin{equation}
 \label{projcon1}
(Y\phi-Z-\xi)'(u-\xi)\leq 0,
\end{equation}
holds for all $u\in \sigma'\Pi $, and
\begin{equation}
 \label{projcon2}
\biggl[\frac{\int_{\R_+}\frac{V}{Y+V}y\lambda\nu(\dy)+b}{\int_{\R_+}\frac{1}{Y+V}y^2\lambda\nu(\dy)}
-\rho\biggr](\alpha-\rho)\leq 0,
\end{equation}
holds for all $\alpha\in\R^{+}$.

The following result is a consequence of the second projection theorem
in~\cite[Theorem 9.8]{Beck}.
\begin{lemma}
\label{proj}
Suppose $\xi\in\sigma'\Pi$, $\rho\geq0$. Then the inequality \eqref{projcon1} holds for all $u\in \sigma'\Pi $ if and only if
\begin{equation}
\label{xi}
\xi=\proj_{\sigma'\Pi}(\phi Y-Z),
\end{equation}
and the inequality \eqref{projcon2} holds for all $\alpha\in\R_+$ if and only if
\begin{equation}
\label{rho}
\rho 
=\frac{\Bigl(\int_{\R_+}\frac{V}{Y+V}y\lambda\nu(\dy)+b\Bigr)^+}{\int_{\R_+}\frac{1}{Y+V}y^2\lambda\nu(\dy)}.
\end{equation}
Here, $\proj_{C}(y)$ denotes the projection of $y$ to a closed convex set
$C$, which is uniquely determined by
\[\vert y-\proj_{C}(y)\vert =\min_{c\in C}\vert y-c\vert .
\] 
\end{lemma}
Taking the above expressions into \eqref{expressf1} yields
\begin{align*}
 f&=-\frac{1}{Y}\Big\vert Z+\frac{\theta}{\Lambda} h\sigma'\hat\pi\Big\vert ^2-\int_{\R_+}\frac{1}{Y+V}\Big\vert V-\frac{\theta}{\Lambda}h\hat qy\Big\vert ^2\lambda\nu(\dy) +\frac{2\theta}{\Lambda} h\hat\pi'\sigma\phi+\frac{2\theta}{\Lambda} h\hat qb\\
&=-\frac{1}{Y}\vert Z+\xi\vert ^2+2\phi'\xi-\int_{\R_+}\frac{1}{Y+V}\vert V-\rho y\vert ^2\lambda\nu(\dy)+2b\rho.
\end{align*}
Overall, we conjecture that
\begin{align}
f&=-\frac{1}{Y}\vert Z+\xi\vert ^2+2\phi'\xi -\int_{\R_+}\frac{1}{Y+V}\vert V-\rho y\vert ^2\lambda\nu(\dy)+2b\rho,\label{funf}\\
 \hat\eta &=-\frac{1}{ Y}\Bigl(Z+\proj_{\sigma'\Pi}(Y\phi-Z)\Bigr), \ \hat\psi=-\frac{1}{Y+V}(V-\rho y),\label{saddleeta}\\
 \hat\pi &=\frac{\Lambda^{\hat\eta,\hat\psi}} {h\theta}(\sigma\sigma')^{-1} \sigma\proj_{\sigma'\Pi}(Y\phi-Z), \ \hat q=\frac{\Lambda^{\hat\eta,\hat\psi}}{h\theta}\rho, \label{hatpi}
\end{align}
where $\xi$ and $\rho$ are given in \eqref{xi} and \eqref{rho} respectively.
Notice that
\begin{align}\label{fexp}
-\frac{1}{Y}\big\vert Z+\xi\big\vert ^2+2\phi'\xi 
=&-\frac{1}{Y}\vert \xi\vert ^2+\frac{2}{Y}(\phi Y-Z)'\xi -\frac{1}{Y}\vert Z\vert ^2\notag\\
=&-\frac{1}{Y} \vert \xi-(\phi Y- Z)\vert ^2+\frac{1}{Y}\vert \phi Y- Z\vert ^2-\frac{1}{ Y}\vert Z\vert ^2\notag\\
=&-\frac{1}{Y}\inf_{ \pi\in\Pi}\vert \sigma' \pi-(\phi Y- Z)\vert ^2
+\frac{1}{Y}\vert \phi Y- Z\vert ^2-\frac{1}{ Y}\vert Z\vert ^2\notag\\
=&-\frac{1}{Y}\inf_{ \pi\in\Pi}\Bigl[\pi'\sigma\sigma' \pi-2\pi'\sigma(\phi Y- Z)\Bigr]-\frac{1}{ Y}\vert Z\vert ^2,
\end{align}
and
\begin{align}
 \int_{\R_+}\frac{1}{Y+V}\Big\vert V-\rho y\Big\vert ^2\lambda\nu(\dy)+2b\rho 
=&\frac{\Bigl[\Bigl(\int_{\R_+}\frac{V}{Y+V}y\lambda\nu(\dy)+b\Bigr)^+\Bigr]^2}
{\int_{\R_+}\frac{1}{Y+V}y^2\lambda\nu(\dy)}-\int_{\R_+}\frac{V^2}{Y+V}\lambda\nu(\dy).
\end{align}
Now we get the desired BSDE \eqref{Y0}:
\begin{equation}
\label{Y} 
\begin{cases} 
\dd Y=\biggm\{\frac{1}{Y}\inf\limits_{ \pi\in\Pi}\Bigl[\pi'\sigma\sigma' \pi-2\pi'\sigma(\phi Y- Z)\Bigr]+\frac{1}{ Y}\vert Z\vert ^2\\
\qquad\qquad- \frac{\Bigl[\Bigl(\int_{\R_+}\frac{V}{Y+V}y\lambda\nu(\dy)+b\Bigr)^+\Bigr]^2}{\int_{\R_+}\frac{1}{Y+V}y^2\lambda\nu(\dy)}+\int_{\R_+}\frac{V^2}{Y+V}\lambda\nu(\dy)\biggm\}\dt\\
\qquad\quad+Z'\dw+\int_{\R_+}V(y)\tilde\gamma(\dt,\dy),\medskip\\
Y_T=1,\quad Y>0, \quad Y+V>0,
\end{cases}\end{equation}
In the next section, we will solve this BSDE and use its solution to derive a solution to the problem \eqref{MMV}.

\section{Solutions to the BSDE \eqref{Y} and the problem \eqref{MMV}}
 \label{Verification}

\begin{definition}
\label{def}
A {triple} $(Y,Z,V)$ is called a solution to the BSDE \eqref{Y} if it satisfies all the equalities and inequalities in \eqref{Y} and $(Y, Z,V)\in S^{\infty}_{\mathbb{F}}(0,T;\mathbb{R})\times L^{2}_{\mathbb{F}}(0,T;\mathbb{R}^m)\times\linnu$. The solution is called uniformly positive if both $Y\geq \delta$ and $Y+V\geq \delta$ a.s. with some deterministic constant $\delta>0$.
\end{definition}

\begin{lemma}
\label{Th:Y}
There exists a unique uniformly positive solution $(Y,Z,V)$ to the BSDE \eqref{Y}.
\end{lemma}
\begin{proof}
Consider (the arguments $t$ and $\omega$ are suppressed)
\begin{equation}
\label{P}
\begin{cases} 
\dd P_t
=-\biggm\{\inf\limits_{\pi\in\Pi}\Bigl[P\pi'\sigma\sigma'\pi-2\pi'(P\sigma\phi+\sigma\Delta)\Bigr] -\frac{\Bigl[\Bigl(Pb-\int_{\R_+}\Gamma(y)y\lambda\nu(\dy)\Bigr)^+\Bigr]^2}{\int_{\R_+}(P+\Gamma(y))y^2\lambda\nu(\dy)}\biggm\}\dt\\
\quad\quad\quad+\Delta'\dw+\int_{\R_+}\Gamma(y)\tilde\gamma(\dt,\dy), \\
P_T=1,\quad P>0, \quad P+\Gamma>0.
\end{cases}\end{equation}
{From~\cite[Theorem 3.3]{SX}, the above BSDE \eqref{P} admits a unique
solution $(P,\Delta,\Gamma)\in S^{\infty}_{\mathbb{F}}(0,T;\mathbb{R})\times L^{2}_{\mathbb{F}}(0,T;\mathbb{R}^m)\times\linnu$ such that $c_1\leq P$, $P+\Gamma\leq c_2$ with some
deterministic constants $c_2>c_1>0$.}
Then
\begin{equation}
\label{PtoY}
(Y,Z,V):= \Bigl(\frac{1}{P}, -\frac{\Delta}{P^2}, \ \frac{1}{P+\Gamma}-\frac{1}{P}\Bigr),
\end{equation}
is well-defined. 
It can be directly verified, using It\^o's formula, that $(Y,Z,V)$ is a solution to \eqref{Y}.
The above change $(P,\Delta,\Gamma)\to (Y,Z,V)$ is invertible, so the uniqueness of uniformly positive solution to \eqref{Y} follows from the fact that \eqref{P} admits a unique solution.
\end{proof}

The following result provides a complete answer to the problem \eqref{MMV}.
\begin{theorem}
\label{solution1}
Let $(Y,Z,V)$ be the unique uniformly positive solution to \eqref{Y}.
Let $h$, $\hat\eta, \hat\psi, \hat\pi, \rho, \hat q$, be defined in \eqref{h}, \eqref{saddleeta}, \eqref{rho}, \eqref{hatpi}, respectively.
Then $(\hat\pi,\hat q, \hat\eta,\hat\psi)$ is a saddle point for
the constrained MMV problem \eqref{MMV} with the optimal value
\begin{equation}
\label{MMVvalue}
\sup_{(\pi,q)\in\mathcal{U}}\inf_{(\eta,\psi)\in\mathcal{A}}\E^{\mathbb{P}^{\eta}}\Bigl[X_T+\frac{1}{2\theta}(\Lambda^\eta_T-1)\Bigr]
=xh_0+\frac{1}{2\theta}(Y_0-1).
\end{equation}
\end{theorem}

\begin{proof}
To prove the theorem, it suffices to prove the following three claims.
\begin{enumerate}
\item $(\hat\pi,\hat q)\in\mathcal{U}$, $(\hat\eta,\hat\psi)\in\mathcal{A}$; 
\item it holds for all $ (\pi,q)\in\mathcal{U}$ that
\begin{equation}
\label{saddle1}
\E^{\mathbb{P}^{\hat\eta,\hat\psi}}\Bigl[X_T^{\pi,q}+\frac{1}{2\theta}(\Lambda^{\hat\eta,\hat\psi}_T-1)\Bigr]
\leq xh_0+\frac{1}{2\theta}(Y_0-1);
\end{equation}
\item it holds for all $(\eta,\psi)\in\mathcal{A}$ that
\begin{equation}
\label{saddle2}
\E^{\mathbb{P}^{\eta,\psi}}\Bigl[X_T^{\hat\pi,\hat q}+\frac{1}{2\theta}(\Lambda^{\eta,\psi}_T-1)\Bigr]
\geq xh_0+\frac{1}{2\theta}(Y_0-1).
\end{equation}
\end{enumerate}
Now we prove these claims one by one.
 \\

\noindent
\textbf{Claim 1.} We have $(\hat\pi,\hat q)\in\mathcal{U}$ and $(\hat\eta,\hat\psi)\in\mathcal{A}$. 

Let $\xi=\proj_{\sigma'\Gamma}(\phi Y-Z)$. 
Then according to Lemma \ref{proj}, \eqref{projcon1} holds.
By taking $u$ as $2\xi$ and $0$ in \eqref{projcon1} respectively, we immediately have
\begin{equation}
\label{projpro}
\vert \xi\vert ^2-\xi'(\phi Y-Z)=-\xi'(\phi Y-Z-\xi)=0.
\end{equation}
Using this and applying It\^o's formula to $\theta h_tX_t^{\hat\pi,\hat q}+Y_t\Lambda_t^{\hat\eta,\hat\psi}$, one can obtain
\[
\dd\;(\theta h_tX_t^{\hat\pi,\hat q}+Y_t\Lambda_t^{\hat\eta,\hat\psi})=0\]
so that
\begin{equation}
\label{identity1}
\theta h_tX_t^{\hat\pi,\hat q}+Y_t\Lambda_t^{\hat\eta,\hat\psi}\equiv\theta h_0 x+Y_0, \ t\in[0,T].
\end{equation}
Since $\xi\in\sigma' \Pi$, there exists $\beta\in\Pi$, such that $\xi=\sigma'\beta$. Let $\rho$ be defined by \eqref{rho}. From \eqref{hatpi},
we get 
\[
\hat\pi=\frac{\Lambda^{\hat\eta,\hat\psi}} {h\theta}(\sigma\sigma')^{-1} \sigma\proj_{\sigma'\Gamma}(Y\phi-Z) 
=\frac{\Lambda^{\hat\eta,\hat\psi}} {h\theta}\beta.
\]
Because $\Pi$ is cone, $\Lambda$, $h$, $\theta>0$ and $\beta\in\Pi$,
we see $\hat\pi\in\Pi$. By \eqref{identity1},
\begin{equation}
\label{hatsigmapi2}
\sigma'_t\hat \pi_t=\frac{\Lambda_t^{\hat\eta,\hat\psi}} {h_t\theta}\sigma'_t\beta_t=\frac{\Lambda_t^{\hat\eta,\hat\psi}} {h_t\theta}\xi_t 
=\frac{\tilde a- h_t X^{\hat\pi,\hat q}_t}{ h_tY_t}\xi_t,
\end{equation} 
where $\tilde a:= h_0x+Y_0/\theta$ is a constant.
Similarly, $\hat q$ can be rewritten as
\begin{equation}
\label{hatq2}
\hat q_t=\frac{\tilde a- h_t X^{\hat\pi,\hat q}_t}{ h_tY_t}\rho_t.
\end{equation}
Substituting \eqref{hatsigmapi2} and \eqref{hatq2} into \eqref{wealth}, and recalling \eqref{h}, we have
\begin{align*}
\dd\;(h X ^{\hat\pi,\hat q} -\tilde a)&=-\frac{1}{Y}(h X ^{\hat\pi,\hat q} -\tilde a)
\Bigl[(\xi'\phi+\rho b)\dt+\xi'\dw-\rho\int_{\R_+}y\tilde\gamma(\dt,\dy)\Bigr].
\end{align*}
Applying It\^o's formula to $\frac{1}{Y}(h X^{\hat\pi,\hat q}-\tilde a)^2$ and using \eqref{funf} and f\eqref{Y}, we have
\begin{align*}
\dd\Bigl[\frac{1}{Y }(h X^{\hat\pi,\hat q} -\tilde a)^2\Bigr]
&=\Bigl\{\frac{\vert \xi\vert ^2}{Y}-2(\xi'\phi+\rho b)-\frac{\vert Z+\xi\vert ^2}{Y}+2\phi'\xi +2b\rho+\frac{\vert Z\vert ^2}{Y}+ \frac{2}{Y}Z'\xi\\
&\quad~+\frac{1}{Y}\int_{\R_+}\Bigl[\rho^2 y^2-\frac{Y}{Y+V}(V-\rho y)^2 +Y^3\Bigl(\frac{1}{Y+V}-\frac{1}{Y}+\frac{V}{Y^2}\Bigr)\\
&\quad~+Y^3\Bigl(\frac{\rho^2y^2}{Y^2}+2\frac{\rho}{Y}y\Bigr)
\Bigl(\frac{1}{Y+V}-\frac{1}{Y}\Bigr)\Bigr]\lambda\nu(\dy)\Bigr\}\frac{1}{Y^2}(h X^{\hat\pi,\hat q} -\tilde a)^2\dt\\
&\quad+(\cdots)\dw+ (\cdots)\tilde\gamma(\dt,\dy).
\end{align*}
Using \eqref{projpro}, a direct calculation shows that the term in $\{\cdots\}$ before $\dt$ equals $0$. Therefore, $\frac{1}{Y}(h X^{\hat\pi,\hat q}-\tilde a)^2$ is a local martingale.
{From now on, let $\tau_n, n = 1,2, \ldots $, be a localizing sequence of stopping times for the local martingales
given by some stochastic integrals with respect to the Brownian motion and the compensated Poisson random measure which may vary in different cases.}
Then
\[
\E\Bigl[\frac{1}{Y_{\iota\wedge\tau_n}}(h_{\iota\wedge\tau_n} X^{\hat\pi,\hat q}_{\iota\wedge\tau_n}-\tilde a)^2\Bigr]=\frac{1}{Y_0}(h_0x-\tilde a)^2=\frac{Y_0}{\theta^2},\]
for any stopping time $\iota\leq T$. 
Sending $n\rightarrow\infty$, it follows from Fatou's Lemma that
\[
\E\Bigl[\frac{1}{Y_{\iota}}(h_{\iota} X^{\hat\pi,\hat q}_{\iota}-\tilde a)^2\Bigr]\leq \frac{Y_0}{\theta^2}.
\] 
Since $c_1\leq Y$, $h\leq c_2$ for some constants $c_2>c_1>0$, we get for some constant $c_3$
\begin{equation}
\label{X2int}
\E\Bigl[ (X^{\hat\pi,\hat q}_{\iota})^2\Bigr]\leq c_3,
\end{equation}
holds for any stopping time $\iota\leq T$. Now it is standard to prove
 $(\hat\pi,\hat q)\in L^2_{\mathcal{F}}(0,T;\mathbb{R}^{m+1})$, see e.g.~\cite[Lemma 3.8]{SX}. This proves $(\hat\pi,\hat q)\in\mathcal{U}$.

From \ref{Th:Y}, there {exist} two positive constants $c_1\leq c_2$ such that
$c_1\leq Y, Y+V, h\leq c_2$. 
Recalling the definition of $\hat\psi$ in \eqref{saddleeta},
\[
\hat\psi_t(y)=-1+\frac{Y_t}{Y_t+V_t(y)}+\rho y\geq -1+\frac{c_1}{c_2}.
\]
As we assumed the size of all claims is uniformly bounded, so $\nu$ is compactly supported, which implies that $\hat\psi\in L^{\infty,\nu}_{\mathcal{P}}(0,T;\R)$.

{From the proof of~\cite[Theorem 3.2, Theorem 3.3]{SX}, $\Delta$, the
second component of the unique solution to \eqref{P}, actually satisfies
\eqref{BMO1}. Moreover, from the fact that $P$ is uniformly positive
and bounded, and the relationship \eqref{PtoY}, we have
\[
\vert \hat\eta\vert =\Big\vert -\frac{1}{ Y}\Bigl(Z+\proj_{\sigma'\Pi}(Y\phi-Z)\Bigr)\Big\vert \leq c_4\vert Z\vert \leq c_5\vert \Delta\vert ,\]
for some positive constants $c_4,c_5$. Therefore $\hat\eta$ also satisfies \eqref{BMO1}.}

On the other hand,
by virtue of \eqref{identity1} and \eqref{X2int}, we immediately get that there exists some positive constant $c_6$,
\begin{equation}
\label{Lambdasqu}
\E[(\Lambda^{\hat\eta,\hat\psi}_{\iota})^2]<c_6,
\end{equation}
for any stopping time $\iota\leq T$. Hence the last requirement in \eqref{def:A} is also satisfied.
This proves $(\hat\eta,\hat\psi)\in\mathcal{A}$, completing the proof of the first claim.
\\

\noindent
\textbf{Claim 2.} The inequality \eqref{saddle1} holds for all $ (\pi,q)\in\mathcal{U}$.

Write
\begin{equation}
\label{defineR1}
R^{(\eta,\psi,\pi,q)}_t=X_t^{\pi,q}h_t+\frac{1}{2\theta}(\Lambda_t^{\eta,\psi} Y_t-1).
\end{equation}
For any $(\pi,q)\in\mathcal{U}$, applying It\^o's formula to $R_{t}^{(\hat\eta,\hat\psi,\pi,q)}$, {and taking expectation with respect to the probability measure $\mathbb{P}^{\hat\eta,\hat\psi}$}, we have
\begin{align}\label{substep2}
\E^{\mathbb{P}^{\hat\eta,\hat\psi}}[R^{(\hat\eta,\hat\psi,\pi,q)}_{T\wedge\tau_n}] 
&= xh_0+\frac{1}{2\theta}(Y_0-1)+\E^{\mathbb{P}^{\hat\eta,\hat\psi}}\int_0^{T\wedge\tau_n}\Bigl\{h\pi'(\sigma\phi + \sigma\hat\eta)\nn\\
&\quad~+hq\Bigl(b-\int_{\R_+}y\hat\psi\lambda\nu(\dy)\Bigr)+\frac{\Lambda}{2\theta}\Bigl[\frac{1}{Y}\vert Z+\xi\vert ^2-2\phi'\xi\nn\\
&\quad~-2b\rho+\hat\eta'Z+\int_{\R_+}\Bigl(\hat\psi V+\frac{(V-\rho y)^2}{Y+V}\Bigr)\lambda\nu(\dy)\nn\\
&\quad~+(Z+Y\hat\eta)'\hat\eta+\int_{\R_+}(V+Y\hat\psi+V\hat\psi)\hat\psi\lambda\nu(\dy)\Bigr]\Bigr\}\dt,
\end{align}
{recalling the stopping times $\tau_n, n = 1,2, \ldots $, are a localizing sequence of the local martingales
given by some stochastic integrals with respect to the Brownian motion and the compensated Poisson random measure.} 
On one hand, for any $\pi\in\Pi$, according to \eqref{projcon1}, \ref{proj} and \eqref{projpro}, we get
\begin{align}
\label{substep1}
h\pi'(\sigma\phi + \sigma\hat\eta)&=\frac{h}{Y}\pi'\sigma(\phi Y-Z-\xi ) \leq \frac{h}{Y}\xi'(\phi Y-Z-\xi)=0.
\end{align}
A similar consideration shows that 
\[
hq\Bigl(b-\int_{\R_+}y\hat\psi\lambda\nu(\dy)\Bigr)\leq 0.
\]
 On the other hand, using \eqref{xi}, \eqref{rho}, \eqref{saddleeta} and \eqref{projpro}, it is not hard to show the term in $\bigl[\cdots\bigr]$ in \eqref{substep2} equals $0$.

Combining above, we arrive at
\begin{align}\label{step2}
&\E^{\mathbb{P}^{\hat\eta,\hat\psi}}[R^{(\hat\eta,\hat\psi,\pi,q)}_{T\wedge\tau_n}]
=\E[\Lambda^{\hat\eta,\hat\psi}_{T\wedge\tau_n}R^{(\hat\eta,\hat\psi,\pi,q)}_{T\wedge\tau_n}] \leq xh_0+\frac{1}{2\theta}(Y_0-1), \ \forall(\pi,q)\in\mathcal{U}.
\end{align} 
For any $(\pi,q)\in\mathcal{U}$, it is standard to verify 
\[
\E\Bigl[\sup_{t\in[0,T]}(X^{\pi,q}_t)^2\Bigr]<\infty.
\]
Since $h$ and $Y$ are bounded, $Y\geq c>0$, we get from \eqref{identity1} and \eqref{defineR1} that
\begin{align*}
\E\Bigl[\sup_{t\in[0,T]}\vert \Lambda^{\hat\eta,\hat\psi}_tR^{\hat\eta,\hat\psi,\pi,q}_t\vert \Bigr] 
&\leq c\E\Bigl[\sup_{t\in[0,T]}(\vert X_t^{\hat\pi, \hat q}\vert +1)(\vert X_t^{\hat\pi, \hat q}\vert +\vert X_t^{\pi, q}\vert +1)\Bigr]\\
&\leq c\E\Bigl[\sup_{t\in[0,T]}\vert X_t^{\hat\pi, \hat q}\vert ^{2}\Bigr]
+c\E\Bigl[\sup_{t\in[0,T]}\vert X_t^{\pi, q}\vert ^{2}\Bigr]+c
<\infty,
\end{align*}
where $c>0$ is a constant that may vary from line to line. 

Sending $n\rightarrow\infty$ in \eqref{step2} and using the dominated convergence theorem, we get \eqref{saddle1}.
\\

\noindent 
\textbf{Claim 3.} The inequality \eqref{saddle2} holds for all $(\eta,\psi)\in\mathcal{A}$.
 
For any $(\eta,\psi)\in\mathcal{A}$, we have
\begin{align}
\E^{\mathbb{P}^{\eta,\psi}}[R^{(\eta,\psi,\hat\pi,\hat q)}_{T\wedge\tau_n}]
 =&xh_0+\frac{1}{2\theta}(Y_0-1)
+\E^{\mathbb{P}^{\eta,\psi}}\int_0^{T\wedge\tau_n}\frac{\Lambda^{\eta,\psi}}{2\theta}
\Bigl[Y\vert \eta\vert ^2+2\eta'Z\nn\\
&~+\frac{2\theta}{\Lambda^{\eta,\psi}}h\hat\pi' \sigma\eta+\frac{2\theta}{\Lambda^{\eta,\psi}}h\hat\pi'\sigma\phi + \frac{2\theta}{\Lambda^{\eta,\psi}}h\hat qb-f\nn\\
&~+\int_{\R_+}\Bigl((Y+V)\psi^2+2\psi V -\frac{2\theta}{\Lambda^{\eta,\psi}}h\hat qy\psi\Bigr)\lambda\nu(\dy)\Bigr]\dt.\label{goon}
\end{align}
Completing the squares with respect to $\eta$ and $\psi$ respectively, the drift term in \eqref{goon} is no less than
\begin{align}
\label{goon2}
&\frac{\Lambda^{\eta,\psi}}{2\theta}\Bigl[-\frac{1}{ Y}\Big\vert Z+\frac{\theta}{\Lambda^{\eta,\psi}}h\sigma'\hat\pi\Big\vert ^2+\frac{2\theta}{\Lambda^{\eta,\psi}}h\hat\pi'\sigma\phi
-\frac{1}{Y+V}\Bigl(V-\frac{\theta}{\Lambda^{\eta,\psi}}h\hat q y\Bigr)^2+\frac{2\theta}{\Lambda^{\eta,\psi}}h\hat qb- f\Bigr].
\end{align}
Please note that $(\hat\pi,\hat q)$ is a feedback of $\Lambda$, so under $(\eta,\psi)$, the optimal $(\hat\pi,\hat q)$ in \eqref{hatpi} should take
\[
\hat\pi =\frac{\Lambda^{\eta,\psi}} {h\theta}(\sigma\sigma')^{-1} \sigma\proj_{\sigma'\Gamma}(Y\phi-Z), \ \hat q=\frac{\Lambda^{\eta,\psi}}{h\theta}\rho.
\]
Recall \eqref{hatsigmapi2}, \eqref{xi} and \eqref{funf}, so \eqref{goon2} is equal to
\begin{align}\label{goon3}
&\frac{\Lambda^{\eta,\psi}}{2\theta}\Bigl[-\frac{1}{Y}\vert Z+\xi\vert ^2+2\phi'\xi -\frac{1}{Y+V}(V-\rho y)^2+2\rho b-f\Bigr]=0.
\end{align}
Hence from \eqref{goon}, \eqref{goon2} and \eqref{goon3}, we have 
\[
\E^{\mathbb{P}^{\eta,\psi}}[R^{(\eta,\psi,\hat\pi,\hat q)}_{T\wedge\tau_n}]\geq xh_0+\frac{1}{2\theta}(Y_0-1).
\]
Similar to the previous argument, we can prove 
\[
\E[\sup_{t\in[0,T]}\vert \Lambda^{\eta,\psi}_tR_t^{(\eta,\psi,\hat\pi,\hat q)}\vert ]<\infty.
\]
 By sending $n\rightarrow\infty$ in above, we get \eqref{saddle2}.
This complete the proof.
\end{proof}

\section{Link to the classical MV problem}
\label{Comparison}
The classical MV problem (see, e.g.~\cite{ZL}) is
\begin{equation}
\label{MV}
\sup_{(\pi,q)\in\mathcal{U}}\Bigl[\E(X_T)-\frac{\theta}{2}\mbox{Var}(X_T)\Bigr].
\end{equation}
We will firstly solve the problem \eqref{MV} and then draw a comparison with the MMV problem \eqref{MMV}.

For any $z\in\mathbb{R}$, let us consider
\begin{equation}
\label{MV1}
F(z):= \inf_{(\pi,q)\in\mathcal{U}^z}\E\Bigl[(X_T-z)^2\Bigr]=\inf_{(\pi,q)\in\mathcal{U}^z}\Bigl[\E(X_T^2)-z^2\Bigr],
\end{equation}
where
\[
\mathcal{U}^z:= \Bigl\{(\pi,q)\in\mathcal{U}\;\Big\vert \;\E(X_T^{\pi,q})=z\Bigr\},\]
with the convention that $\inf\emptyset =+\infty$. {By definition}, $F(z)\geq0$.

As illustrated in~\cite[page 12]{TZ},
the connection between problems \eqref{MV} and \eqref{MV1} is given as follows:
\begin{align}
\sup_{(\pi,q)\in\mathcal{U}}\Bigl[\E(X_T)-\frac{\theta}{2}\mbox{Var}(X_T)\Bigr] 
&=\sup_{z\in\mathbb{R}}\sup_{(\pi,q)\in\mathcal{U}^z}\Bigl[z-\frac{\theta}{2}\Bigl(E(X_T^2)-z^2\Bigr)\Bigr]\nn\\
&=\sup_{z\in\mathbb{R}}\Bigl[z-\frac{\theta}{2}\inf_{(\pi,q)\in\mathcal{U}^z}\Bigl(E(X_T^2)-z^2\Bigr)\Bigr]=\sup_{z\in\mathbb{R}}\Bigl[z-\frac{\theta}{2} F(z)\Bigr].\label{Lag0}
\end{align}

In order to solve \eqref{MV1}, we introduce a Lagrange multiplier $\zeta\in\mathbb{R}$ and consider
\begin{align}
\label{MV2}
J(z,\zeta)&:= \inf_{(\pi,q)\in\mathcal{U}}\E\Bigl[{X_T^2-z^2-2\zeta(X_T-z)}\Bigr] =\inf_{(\pi,q)\in\mathcal{U}}\Bigl[\E(X_T-\zeta)^2-(z-\zeta)^2\Bigr].
\end{align}
By the Lagrange duality theorem (see Luenberger~\cite{Lu}),
\begin{equation}
\label{Lag1}
F(z)=\sup_{\zeta\in\mathbb{R}}J(z,\zeta), \ z\in\mathbb{R}.
\end{equation}

According to~\cite{SX},
the solution for problem \eqref{MV2} depends on the following two coupled BSDEs:
\begin{equation}
\label{P1}
\begin{cases}
\dd P_1=-\Bigl[2rP_1+F_1^{*}(t,P_1,\Delta_{1})+G_1^{*}(P_1,\Gamma_1,P_2,\Gamma_2)\Bigr]\dt +\Delta_{1}'\dw+\int_{\R_+}\Gamma_1(y)\tilde \gamma(\dt,\dy), \\
P_{1,T}=1, \ P_{1}>0, \ P_{1}+\Gamma_{1}>0, \ P_{2}+\Gamma_{2}>0,
\end{cases}\end{equation}
and
\begin{equation}
\label{P2}
\begin{cases}
\dd P_2=-\Bigl[2rP_2+F_2^{*}(t,P_2,\Delta_{2})+G_2^{*}(P_2,\Gamma_2)\Bigr]\dt +\Delta_{2}'\dw+\int_{\R_+}\Gamma_2(y)\tilde \gamma(\dt,\dy), \\
P_{2,T}=1, \ P_{2}>0, \ \ P_{2}+\Gamma_{2}>0,
\end{cases}\end{equation}
where, 
for $(t,\omega,v,u,P_{1},\Delta_{1},\Gamma_1,P_{2},\Delta_{2},\Gamma_2)\in[0,T]\times\Omega\times\Pi\times\R_+\times
\R_+\times\R^m\times \lnu\times\R_+\times\R^m\times \lnu$, we define
\begin{align*} 
F_1(t,\omega,v,P_1,\Delta_1)&:= P_1\vert \sigma'_tv\vert ^2+2v'(P_1\mu_t+\sigma_t\Delta_1),\\
F_2(t,\omega,v,P_2,\Delta_2)&:= P_2\vert \sigma'_tv\vert ^2-2v'(P_2\mu_t+\sigma_t\Delta_2),\\
F_1^{*}(t,\omega,P_1,\Delta_1)&:= \inf_{v\in\Pi}F_1(t,\omega,v,P_1,\Delta_1),\\
F_2^{*}(t,\omega,P_2,\Delta_2)&:= \inf_{v\in\Pi}F_2(t,\omega,v,P_2,\Delta_2),
\end{align*}
and
\begin{align*}
&G_1^{*}(P_1,\Gamma_1,P_2,\Gamma_2):= \inf_{u\geq0}G_1(u,P_1,\Gamma_1,P_2,\Gamma_2), \\
&G_2^{*}(P_2,\Gamma_2)
:= -\frac{\Bigl[\Bigl(P_2 b -\int_{\R_+}\Gamma_2(y) y\lambda\nu(\dy)\Bigr)^+\Bigr]^2}
{\int_{\R_+} (P_2+\Gamma_2(y))y^2\lambda\nu(\dy)},
\end{align*}
and
\begin{align*}
G_1(u,P_1,\Gamma_1,P_2,\Gamma_2)
:=& \int_{\R_+}\Bigl[(P_1+\Gamma_1(y))\bigl[[(1-uy)^+]^2-1\bigr]\\
&\quad+(P_2+\Gamma_2(y))[(1-uy)^-]^2\Bigr]\lambda \nu(\dy)+2uP_1(b+\lambda b_Y).
\end{align*}

The definitions of solutions to \eqref{P1} and \eqref{P2} agree with the one in \ref{def}. Recall that $h_t=e^{\int_t^Tr_s\ds}$ is deterministic and bounded.
From~\cite[Theorem 3.2, Theorem 3.3 and Lemma 4.1]{SX}, there exists a unique
uniformly positive solution $(P_1,\Delta_1,\Gamma_1,P_2,\Delta_2,\Gamma_2) $ to the coupled BSDEs \eqref{P1} and \eqref{P2};
moreover, one has
$P_{1,0}\leq h_0^2$ and $P_{2,0}< h_0^2$.

The following result for the problem \eqref{MV2} comes from~\cite[Theorem 3.7]{SX}. We present it here in terms of our notation.
\begin{lemma}
\label{optimalMV1}
Let $(P_1,\Delta_1,\Gamma_1)$ and $(P_2,\Delta_2,\Gamma_2)$ be the unique uniformly positive solutions to \eqref{P1} and \eqref{P2}, respectively.
Let
\begin{align*}
\xi_{1,t}&:= (\sigma_t\sigma_t')^{-1}\sigma_t Proj_{\sigma'\Gamma}\Bigl(-\phi_t-\frac{\Delta_{1,t}}{P_{1,t}}\Bigr), \\
\xi_{2,t} &:= (\sigma_t\sigma_t')^{-1}\sigma_t Proj_{\sigma'\Gamma}\Bigl(\phi_t+\frac{\Delta_{2,t}}{P_{2,t}}\Bigr),\\ 
\rho_{1,t}&:= \argmin_{u\geq0} G_1(u,P_{1,t},\Gamma_{1,t},P_{2,t},\Gamma_{2,t}),\\
\rho_{2,t} &:= \frac{\Bigl(P_{2,t}b-\int_{\R_+}\Gamma_{2,t}(y)y\lambda\nu(\dy)\Bigr)^+}
{\int_{\R_+}(P_{2,t}+\Gamma_{2,t}(y))y^2\lambda\nu(\dy)}. 
\end{align*}
Then the pair of feedback investment-reinsurance strategies
\begin{align}
\pi^\zeta(t, X)&=\Bigl(X_t-\frac{\zeta}{h_t}\Bigr)^+ \xi_{1,t} +\Bigl(X_t-\frac{\zeta} {h_t}\Bigr)^-\xi_{2,t},\label{pistar}\\
q^{\zeta}(t,X)&=\Bigl(X_t-\frac{\zeta}{h_t}\Bigr)^+ \rho_{1,t} +\Bigl(X_t-\frac{\zeta} {h_t}\Bigr)^-\rho_{2,t}
\end{align}
is optimal for the problem \eqref{MV2} with the optimal value 
\[
J(z,\zeta)=P_{1,0}\Bigl[\Bigl(x-\frac{\zeta}{h_0}\Bigr)^+\Bigr]^2
+P_{2,0}\Bigl[\Bigl(x-\frac{\zeta}{h_0}\Bigr)^-\Bigr]^2-(\zeta-z)^2.
\]
\end{lemma}

To solve \eqref{Lag1},
we do a tedious calculation and obtain
\[
F(z)
=
\begin{cases}
\frac{P_{2,0}(z-xh_0)^2}{h_0^2-P_{2,0}}, & \ \mbox{if} \ z> xh_0 \ \mbox{and} \ P_{2,0}<h_0^2;\\
+\infty, & \ \mbox{if} \ z<xh_0 \ \mbox{and} \ P_{1,0}=h_0^2;\\
\frac{P_{1,0}(z-xh_0)^2}{h_0^2-P_{1,0}}, & \ \mbox{if} \ z<xh_0 \ \mbox{and} \ P_{1,0}<h_0^2;\\
0, & \ \mbox{if} \ z=xh_0,
\end{cases}\]
with the argument maximum
\begin{equation}
\label{hatgammaK}
\hat\zeta(z)=
\begin{cases}
\frac{h_0^2z-xP_{2,0}h_0}{h_0^2-P_{2,0}}, & \ \mbox{if} \ z> xh_0 \ \mbox{and} \ P_{2,0}<h_0^2;\\
-\infty, & \ \mbox{if} \ z<xh_0 \ \mbox{and} \ P_{1,0}=h_0^2;\\
\frac{h_0^2z-xP_{1,0}h_0}{h_0^2-P_{1,0}}, & \ \mbox{if} \ z<xh_0 \ \mbox{and} \ P_{1,0}<h_0^2;\\
xh_0, & \ \mbox{if} \ z=xh_0.
\end{cases}\end{equation}
Accordingly,
\begin{equation}
\label{MVvalue2}
\sup_{z\in\mathbb{R}}\Bigl[z-\frac{\theta}{2}F(z)\Bigr]
=xh_0+\frac{1}{2\theta}\Bigl(\frac{h_0^2}{P_{2,0}}-1\Bigr),
\end{equation}
with the argument maximum
\begin{equation}
\label{hatK}
\hat z= x h_0+\frac{1}{\theta}\Bigl(\frac{h_0^2}{P_{2,0}}-1\Bigr)> xh_0.
\end{equation}
Substituting \eqref{hatK} into \eqref{hatgammaK}, we obtain
\[ 
\hat\zeta(\hat z)=xh_0+\frac{1}{\theta}\frac{h_0^2}{P_{2,0}}.
\] 

The above analysis leads to the following results for the problem \eqref{MV}.
\begin{theorem}
Use the notations in \ref{optimalMV1}.
Define a constant
\begin{equation} 
\hat\zeta=xh_0+\frac{1}{\theta}\frac{h_0^2}{P_{2,0}},
\end{equation}
and {a pair of feedback strategies}
\begin{equation}
\label{hatpimv}
\pi^{\hat\zeta}(t,X)=-\Bigl( X-\frac{\hat\zeta}{h_t}\Bigr)\xi_{2,t}, \ q^{\hat\zeta}=-\Bigl(X-\frac{\zeta} {h_t}\Bigr)\rho_{2,t}.
\end{equation}
Then $(\pi^{\hat\zeta},q^{\hat\zeta})$ is the optimal feedback investment-reinsurance strategy for the problem \eqref{MV} with the optimal value
\begin{equation}
\label{MVvalue}
\sup_{(\pi,q)\in\mathcal{U}}\Bigl[\E(X_T)-\frac{\theta}{2}\mbox{Var}(X_T)\Bigr]
=xh_0+\frac{1}{2\theta}\Bigl(\frac{h_0^2}{P_{2,0}}-1\Bigr).
\end{equation}
\end{theorem}
\begin{proof}
Notice $\hat\zeta=\hat\zeta(\hat z)$; so, by \eqref{Lag0} and \eqref{Lag1},
\begin{align*}
\sup_{(\pi,q)\in\mathcal{U}}\Bigl[\E(X_T)-\frac{\theta}{2}\mbox{Var}(X_T)\Bigr]
&=\sup_{z\in\mathbb{R}}\Bigl[z-\frac{\theta}{2} F(z)\ddp{)}\Bigr]\\
&=\hat z-\frac{\theta}{2} F(\hat z)
=\hat z-\frac{\theta}{2} J(\hat z,\hat\zeta(\hat z))
=\hat z-\frac{\theta}{2} J(\hat z,\hat\zeta).
\end{align*}
Therefore, the optimal portfolio for the problem \eqref{MV2} with $(z,\zeta)=(\hat z,\hat\zeta)$ is also optimal to \eqref{MV}.

We now show that the strategy \eqref{hatpimv} is optimal to the problem \eqref{MV2} with $(z,\zeta)=(\hat z,\hat\zeta)$.
Taking \eqref{hatpimv} to the wealth process \eqref{wealth}, we get
\[
\begin{cases}
\dd\Bigl(X_t^{\pi^{\hat\zeta}}-\frac{\hat\zeta}{h_t}\Bigr)=\Bigl(X_t^{\pi^{\hat\zeta}}-\frac{\hat\zeta}{h_t}\Bigr)
\Bigm((r-\xi_2'\mu-b\rho_2)\dt -\xi_2\sigma'\dw+\int_{\R_+}\rho_2 y\tilde\gamma(\dt,\dy)\Bigm),\\
X^{\pi^{\hat\zeta}}_0-\frac{\hat\zeta}{h_0}=-\frac{h_0}{\theta P_{2,0}},
\end{cases}
\] 
Since $-\frac{h_0}{\theta P_{2,0}}\leq 0$ and $\rho_2\geq0$, it follows that $X_t^{\pi^{\hat\zeta}}\leq\frac{\hat\zeta}{h_t}$.
Accordingly, we can rewrite \eqref{hatpimv} as
\[
\pi^{\hat\zeta}(t, X) 
=\Bigl(X_t-\frac{\hat\zeta}{h_t}\Bigr)^+ \xi_{1,t} +\Bigl(X_t-\frac{\hat\zeta} {h_t}\Bigr)^-\xi_{2,t},\]
and
\[
q^{\hat\zeta}(t, X) 
=\Bigl(X_t-\frac{\hat\zeta}{h_t}\Bigr)^+ \rho_{1,t} +\Bigl(X_t-\frac{\hat\zeta} {h_t}\Bigr)^-\rho_{2,t},\]
which is just \eqref{pistar}, the optimal strategy for the problem \eqref{MV2} with $(z,\zeta)=(\hat z,\hat\zeta)$.

Finally, \eqref{MVvalue} comes from \eqref{Lag0} and \eqref{MVvalue2} evidently.
\end{proof}

In the end, we have the following connection between the problems \eqref{MMV} and \eqref{MV}.
\begin{theorem}
\label{com}
The problems \eqref{MMV} and \eqref{MV} share the same optimal value $xh_0+\frac{1}{2\theta}(Y_0-1)$ and the same optimal feedback portfolio
$\hat\pi=\pi^{\hat\zeta}$,
where $\hat\pi$ and $\pi^{\hat\zeta}$ are defined in \eqref{hatpi} and \eqref{hatpimv}, respectively.
\end{theorem}
\begin{proof}
Recall that $(Y,Z,V)$ is the unique uniformly positive solution to \eqref{Y}, and $h_t=e^{\int_t^Tr_s\ds}$. It can be directly verified, by It\^o's formula, that
\begin{equation}
\label{YtoP}
(Y,Z,V)=\Bigl(\frac{h^2}{P_2}, -\frac{h^2}{P_2^2}\Delta_2, \ \frac{h^2}{P_2+\Gamma_2}-\frac{h^2}{P_2}\Bigr),
\end{equation}
where $(P_2,\Lambda_2,\Lambda_2)$ is the unique uniformly positive solution to \eqref{P2}.
Comparing \eqref{MMVvalue} and \eqref{MVvalue}, the first claim follows.

By \eqref{hatsigmapi2}, \eqref{hatq2}, and \eqref{YtoP}, we have
\begin{align*}
\hat\pi&=\frac{xh_0+\frac{Y_0}{\theta}- h X^{\hat\pi,\hat q} }{ h Y }(\sigma\sigma')^{-1}\sigma \proj_{\sigma'\Gamma}(\phi Y-Z)\\ 
&=\frac{\hat\zeta- h X^{\hat\pi,\hat q} }{ h }(\sigma\sigma')^{-1}\sigma \proj_{\sigma'\Gamma}\Bigl(\phi +\frac{\Delta_2}{P_2}\Bigr) =-\Bigl( X ^{\hat\pi,\hat q}-\frac{\hat\zeta}{h }\Bigr)\xi_2,
\end{align*}
and
\begin{align*}
\hat q&=\frac{xh_0+\frac{Y_0}{\theta}- h X^{\hat\pi,\hat q} }{ h Y }\frac{\Bigl(\int_{\R_+}\frac{V}{Y+V}y\lambda\nu(\dy)+b\Bigr)^+}{\int_{\R_+}\frac{1}{Y+V}y^2\lambda\nu(\dy)}\\
&=\frac{xh_0+\frac{h_0^2}{\theta P_{2,0}}- h X^{\hat\pi,\hat q} }{ h ^2}P_2\frac{-\Bigl(\int_{\R_+}\frac{\Gamma_2}{P_2}y\lambda\nu(\dy)+b\Bigr)^+}{\int_{\R_+}(P_2+\Gamma_2)y^2\lambda\nu(\dy)} =-\Bigl( X ^{\hat\pi,\hat q}-\frac{\hat\zeta}{h }\Bigr)\rho_2.
\end{align*}
which is exactly the feedback optimal strategy \eqref{hatpimv}.
\end{proof}

\section{Concluding remarks}
\label{conclude}
In this paper, we studied the MMV optimal investment-reinsurance problem under the Cram\'er-Lundberg model with random coefficients and strategy constraints. {Semi-closed-form solutions}, the optimal probability measure and optimal investment-reinsurance in terms of the solutions to some BSDEs with jumps are provided. Eventually, we conclude that the solution coincides with the corresponding optimal investment-reinsurance problem under the classical MV criterion.

In the definition \eqref{def:A}, the condition $\psi\geq -1+\epsilon$ ensures that the Dol\'eans-Dade stochastic exponential $\Lambda^{\eta,\psi}$ is strictly positive. The properties of the solution $(Y,Z,V)$ to BSDE with jumps \eqref{Y} help to ensure that the optimal $\hat\psi$ defined in \eqref{saddleeta} fits exactly the subtle point $\hat\psi\geq -1+\epsilon$. 
It is critical to notice that our model only considers down side jumps since they are coming from claims. It is interesting to consider general MMV problems with two-side Poisson jumps, in which case the optimal $\hat\psi$ may not be easy to verify that it satisfies $\hat\psi\geq -1+\epsilon$.

\end{document}